\documentclass[conference,a4paper]{IEEEtran}
\usepackage{epsfig}
\usepackage{times}
\usepackage{float}
\usepackage{afterpage}
\usepackage{amsmath}
\usepackage{mathrsfs}
\usepackage{amstext}
\usepackage{amssymb,bm}
\usepackage{latexsym}
\usepackage{color}
\usepackage{graphicx}
\usepackage{amsmath}
\usepackage{amsthm}
\usepackage{graphicx}
\usepackage[center]{caption}
\usepackage{pstricks}
\usepackage{subfigure}
\usepackage{booktabs}
\usepackage{multicol}
\usepackage{lipsum}
\usepackage{dblfloatfix}
\usepackage{algorithm} 
\usepackage[noend]{algpseudocode} 
\usepackage[normalem]{ulem}
\usepackage{enumitem}

\newcommand{\mcal}{\mathcal}

\newcommand{\mbb}{\mathbb}

\newtheorem{thm}{Theorem}

\newtheorem{prop}[thm]{Proposition}
\newtheorem{rem}{Remark}

\newtheorem*{defin*}{Definition}

\algnewcommand\algorithmicforeach{\textbf{for each}}

\algdef{S}[FOR]{ForEach}[1]{\algorithmicforeach\ #1\ \algorithmicdo}

\begin{document}

\title{Communication vs Distributed Computation:\\ an alternative trade-off curve} 
\author{
\IEEEauthorblockN{Yahya H. Ezzeldin, Mohammed Karmoose, Christina Fragouli}
 University of California, Los Angeles, CA 90095, USA,\\
Email: \{yahya.ezzeldin, mkarmoose, christina.fragouli\}@ucla.edu\\
}
\IEEEoverridecommandlockouts
\maketitle

\begin{abstract}
  In this paper, we revisit the communication vs. distributed computing trade-off, studied within the framework of MapReduce in \cite{li2016fundamental}. 
  An implicit assumption in the aforementioned work is that each server performs all possible computations on all the files stored in its memory. 
  Our starting observation is that, if servers can compute only the intermediate values they need, then storage constraints do not directly imply computation constraints. 
  We examine how this affects the communication-computation trade-off and suggest that the trade-off be studied with a predetermined storage constraint. 
  We then proceed to examine the case where servers need to perform computationally intensive tasks, and may not have sufficient time to perform all computations required by the scheme in~\cite{li2016fundamental}.
  Given a threshold that limits the computational load, we derive a lower bound on the associated communication load, and propose a heuristic scheme that achieves in some cases the lower bound. 

\end{abstract}

\section{Introduction}
Distributed computation across a set of wireless networked servers is well motivated for several practical constraints: we may want to speed up computation time so as to finish a computation faster; we may have partial view of the files needed for computation across servers; we may have limited memory in each server;  or we may be motivated by energy constraints. In this paper we consider the  distributed computing framework studied in \cite{li2016fundamental}, that follows the architecture of MapReduce~\cite{dean2008mapreduce}. 

Our starting observation is that, the system in \cite{li2016fundamental} does not explicitly separate computation from storage. 
The  system uses a cluster of $K$ servers to compute $Q$ output functions from $N$ input files. 
Each file is stored in $r$ different servers, balancing the amount of storage across servers. 
The work in \cite{li2016fundamental} calculates the trade-off 
between the amount of computation and communication that servers need to do for such file placement. 
However, an underlying assumption of the derived trade-off, is that each server performs all possible computations on all the files stored in its memory. 
It is natural to ask: is it indeed useful to perform all possible computations?

The following simple example illustrates that this is not always the case. 
 Consider a cluster with $K{=}3$ servers, $N{=}3$ files and $Q{=}3$ output functions. 
    All $3$ files are available at each server and each server is required to compute only one of the output functions. 
    In this case, instead of performing $9$ computations per server (as assumed in~\cite{li2016fundamental}), each server only needs to perform computations related to its dedicated output function, i.e., only 3 computations are needed per server. 

    Our first contribution is to generalize this observation and derive an alternative trade-off curve to the scheme in \cite{li2016fundamental}. 
    We explicitly use three parameters: $C_{total}$ the total amount of computation required; $r$ that captures the memory requirements; and the communication load $L$. 
    We consider the placement and communication scheme in \cite{li2016fundamental}, and calculate the minimum number of computations each server needs to perform. {We take into account} the amount computed by the server for its assigned output functions, the amount that need to be communicated to other servers, and the amount needed to use as side information to decode transmissions from other servers.

    We then proceed to  examine the case where servers need to perform computationally intensive tasks, and in particular, do not have sufficient time to perform all computations the curve in \cite{li2016fundamental} requires. Such a scenario may occur in wireless, where we may have cheap mobile devices with low computational power that need to cooperatively perform time-critical operations, for scientific computing or virtual reality applications. We ask, if the cluster is limited to perform an amount of computation below a  threshold, what is the resulting minimum communication required to achieve the function computation. 

  Our second contribution is to derive a lower bound for the communication-computation trade-off when a cluster has a limited computation budget. 
    For this lower bound, we assume that the files are distributed across the cluster with a predetermined level of redundancy that does not grow with the available computation budget.
    We show that a scheme directly inferred from~\cite{li2016fundamental} performs poorly when compared against the derived lower bound.
    Finally, we develop a distributed computing scheme inspired by~\cite{li2016fundamental} and show through numerical evaluation that the communication-computation trade-off it provides is comparable to the aforementioned lower bound.
    

\noindent {\bf Related Work.}
Minimizing communication load for distributed computation tasks has received considerable attention in the literature: starting from distributed boolean function computation between two parties \cite{yao1979some,orlitsky2001coding} to the more generalized theory of {\it communication complexity} \cite{kushilevitz2006communication,becker1998communication}. A key concept in reducing the needed amount of communication is through network coding. 
A prominent example of this concept is in the context of distributed cache networks~\cite{maddah2014fundamental,karamchandani2016hierarchical,hachem2015content}, where coding is used in either the data placement or data delivery phases to reduce the amount of communication in the delivery phase.
Recently, coding was also considered in the context of distributed computing systems that are based on the MapReduce framework \cite{li2016fundamental,li2016coded,yu2017optimally}. In fact, the authors in \cite{li2016fundamental} provided a Coded Distributed Computing (CDC) scheme which reduces the amount of communication needed in the data shuffling phase by using coded multicast transmissions.
Our work differs in that we separate computation and storage, and thus derive alternative trade-off curves depending on the relative values of these parameters.
%
%

\section{System Model}\label{sec:system}
\noindent{\bf Notation.} Calligraphic letters denote sets through out the paper. $|\mcal{A}|$ denotes the cardinality of the set $\mcal{A}$. 
The expression $[a:b]$ denotes the set of integers from $a$ to $b$.

\medskip


\noindent {\bf MapReduce framework.}
We consider a cluster of $K$ servers that computes $Q$ output functions  $\phi_q$, $q \in [1:Q]$, from $N$ input files $w_n$, $n \in [1:N]$. 
In this paper, we assume that the servers share a lossless broadcast domain: a transmission from a server can be losslesly received by all other servers.

We assume the cluster uses a MapReduce framework to compute the set of $Q$ functions in a distributed manner. MapReduce is based on the assumption that each output function can be calculated as a function of some intermediate processing of the files. In other words,
$    \phi_q(w_1,\dots,w_N{)}=  h_q(v_{q,1},\dots,v_{q,N})$,
where $v_{q,n} = g_{q,n}(w_n)$ is the intermediate value computed from file $w_n$ relevant to the output function~$\phi_q$, and has length $T$ bits.
In MapReduce terminology, the intermediate value is computed (or ``mapped'') using a map function $g_{q,n}$  and $h_q$ ``reduces'' the intermediate values $\{v_{q,n}\}_{n=1}^N$ to output $\phi_q$.

Based on this decomposition, the computation model in \cite{li2016fundamental} consists of three phases: \emph{Map}, \emph{Shuffle} and \emph{Reduce}. Additionally, a \emph{Placement} phase distributes files and tasks among the servers in the cluster. We next describe each of the phases:

\noindent 1) {\em Placement Phase:} Each server $k$ is loaded with a subset $\mcal{M}_k$ of the $N$ files, such that $\cup_k \mcal{M}_k = [1:N]$.
Each server $k$ is also assigned to compute a partition $\mcal{W}_k$ of the output functions, where $\cup_k \mcal{W}_k = [1:Q]$.

\noindent 2) {\em Map Phase:} Each server $k$ computes a subset $\mcal{C}_k$ of the intermediate values related to $\mcal{M}_k$, i.e., $\mcal{C}_k \subseteq \{v_{q,n} | q \in [1:Q], n \in \mcal{M}_k \}$.
At the end of the Map phase, the assigned computation subsets satisfy that $\cup_k \mcal{C}_k = \{v_{q,n} | q \in [1:Q], n \in [1:N]\}$.  
\begin{rem}
        In MapReduce, files are mapped by presenting them as $(key, value)$ pairs to a $map(\cdot)$ function that outputs a set of intermediate $(key,value)$ pairs based on the input pair. 
        Although, the same $map(\cdot)$ build is used across the servers, the function can output different sets intermediate values based on the server ID by including this information in the $key$.
\end{rem}
\noindent 3) {\em  Shuffle Phase:} For a server $k$ to compute a function $\phi_q$ where $q \in \mcal{W}_k$, it needs all the intermediate messages $\mcal{V}_{q} = \left\{v_{q,n} | q \in \mcal{W}_k, n \in [1{:}N] \right\}$. 
Thus in the Shuffle phase, the $K$ servers exchange intermediate values, such that each server has access to all its needed sets $\mcal{V}_q$. 
The shuffling scheme can be described as follows: each server $k$ creates a message $X_k$ that is a function of its locally computed intermediate values and broadcasts this message $X_k$ to the remaining $K-1$ nodes.

\noindent 4)  {\em Reduce Phase:} In the Reduce phase, server $k$ uses its locally computed intermediate values and the received transmissions $X_1,\dots,X_K$ to decode the set of the needed intermediate values $\mcal{V}_q$, $\forall q \in \mcal{W}_k$.  Using $\mcal{V}_q$, {the nodes} can now compute the desired {functions} $\phi_q = h_q(\mcal{V}_q)$, $\forall q \in {\mcal{W}_k}$. 

\medskip

\noindent{\bf Performance metrics.}
We measure the performance of this computation cluster across three parameters: the \emph{load redundancy} ($r$), the \emph{computation load} ($C_{total}$) and the \emph{communication load} ($L$), defined as follows:\\
\noindent$\bf \bullet$ {\em Load Redundancy.} {\rm We define the \emph{load redundancy} as the average number of times a file is assigned across the servers. We denote this by $r$,  i.e., $r \triangleq \frac{\sum_{k=1}^K |\mcal{M}_k|}{N}$.} Load redundancy captures memory constraints. \\ 
\noindent$\bf \bullet$ {\em Computation Load.} We define the \emph{computation load}  $C_{total}\triangleq\sum_k |\mcal{C}_k|$ as the total number of computations performed across servers in the cluster. \\ 
\noindent$\bf \bullet$ {\em Communication Load.} We define the \emph{communication load} $L \triangleq \sum_{k = 1}^K \frac{b(X_i)}{QNT}$, as the number of bits transmitted in the Shuffle phase normalized by $QNT$,
where $b(X_i)$ is the number of bits used to represent $X_i$ and $QNT$ is the total number of bits in all intermediate values $v_{q,n}$, for $q = [1:Q]$ and $n = [1:N]$.
From the definition, we have $0 \leq L \leq 1$.\\
The definitions of $L$ and $r$ follow  \cite{li2016fundamental}; however
in this paper, we explicitly separate the redundancy from the computation load, and use different parameters for each. 

\section{On the relation between redundancy and computation}
\label{sec:comp_vs_red}
An underlying assumption in  \cite{li2016fundamental} is that each server $k$ must compute all the intermediate value for its stored files $\mcal{M}_k$. In other words, $\mcal{C}_k = \{v_{q,n} | q \in [1:Q], n \in \mcal{M}_k \}$.  In this case, the load redundancy $r$ is linearly proportional to the total number of computations in the system as $|\mcal{C}_k| = |\mcal{M}_k| Q$ and $r$ can be therefore regarded as the computation load.
However, if the server can selectively choose which intermediate values of $\mcal{M}_k$ to compute in the Map phase (as long as the communication load is the same), then the total number of computations is not necessarily linearly correlated with $r$.

Consequently, an increase in $r$ does not necessarily result in an increase in the number of computations performed by the cluster. For example, assume that  
$Q=K$  and each server is required to compute 1 output function (without loss of generality, $\mcal{W}_k = \{k\}$). 
Then, we have $C_{total} = NQ$ for both $r=1$ and $r=K$. 
For $r=1$, each file is available at only one server, thus each server needs to compute all intermediate values for all files stored in its memory. 
For $r=K$, all files are available at each server. 
Thus, each server needs only to compute $N$ intermediate values related to its output function. 
In both cases, the optimal communication load $L(r) = \frac{K-r}{rK}$ is achieved~\cite[Theorem 1]{li2016fundamental}.
Note that for $r=K$, if the servers computed all intermediate values for their files, there would be $NQK$ computations instead of $NQ$.

Later in this section, we characterize the minimum computation load needed by the Coded Distributed Computing (CDC) scheme in~\cite{li2016fundamental} in order to achieve the optimal communication load $L^\star(r)$ in \cite[Theorem 1]{li2016fundamental} for $r = [1:K]$.
As we see later, taking this minimum computation load into account changes the trade-off in~\cite{li2016fundamental} for CDC.
As a preliminary to that discussion, we next briefly describe the CDC scheme in \cite{li2016fundamental}.

\medskip

\noindent {\bf An overview of the  CDC scheme.}\label{sec:CDC}
Assume that $N$ and $Q$ are sufficiently large so that $N = {K \choose r}\eta_1$ and $Q = K\eta_2$ for some $\eta_1,\eta_2 \in \mbb{N}$.
The CDC scheme operates as follows (see \cite{li2016fundamental} for a complete description):

\noindent 1) \emph{Placement Phase}: A disjoint subset $\mcal{M}_\mcal{T}$ of the files is assigned to each subset  $\mcal{T}$ of $r$ servers where $|\mcal{M}_{\mcal{T}}| = \eta_1$. Every server is thus assigned a set of $rN/K$ files and every $\eta_1$ partition of these files is shared with a unique set of $r-1$ other servers.
Every server $k$ is also assigned a unique subset $\mcal{W}_k$ of the output functions to calculate such that $|\mcal{W}_k| = \eta_2$ .\\
\noindent 2) \emph{Map Phase}: Every server computes all possible  intermediate function values for  the files it has.\\
\noindent 3) \emph{Shuffling Phase}: The shuffling phase repeats the following procedure for every set $\mcal{S} \subseteq [1:K]$ of size $r+1$:

{\noindent (i)} For every $i\!\in\!\mcal{S}$, define $\mcal{S}_i \!=\! \mcal{S}\backslash \{i\}$ and identify $\mcal{V}_{\mcal{S}_i}^{i}$ as
\begin{align}
    \mcal{V}_{\mcal{S}_i}^i \triangleq \{ v_{q,n} | n \in \cap_{j \in \mcal{S}_i} \mcal{M}_{j},\ q \in \mcal{W}_i \}.
    \label{eq:V_S_S1}
\end{align}
The set $\mcal{V}_{\mcal{S}_i}^i$ represents the intermediate values that are needed by server $i$ to compute functions in $\mcal{W}_i$, which can be computed {\it exclusively} by all servers in $\mcal{S}_i$ (recall that a file is replicated at exactly $r$ servers). 
Note that $|\mcal{V}_{\mcal{S}_i}^i| \!=\! \eta_1\eta_2$.

{\noindent (ii)} Split every intermediate value in $\mcal{V}_{\mcal{S}_i}^i$ into $r$ disjoint parts of $T/r$ bits and associate each part with a server in $\mcal{S}_i$. 
Thus we split the set $\mcal{V}_{\mcal{S}_i}^i$ into $r$ partitions denoted by $\mcal{V}_{\mcal{S}_i,j}^i$, $j \in \mcal{S}_i$,  each of size $\eta_1\eta_2 \frac{T}{r}$. Each server $j$ will be responsible to convey its part to server $i$ with coded broadcast transmissions.

{\noindent (iii)} After splitting all sets $\mcal{V}_{\mcal{S}_i}^i$ for all $i \in \mcal{S}$ (we have $r+1$ such sets), server $k$ sends the bit-wise XOR of all the $\eta_1\eta_2\frac{T}{r}$-bit parts in $\mcal{U}_k^\mcal{S} \triangleq \bigcup_{i \in \mcal{S}} \mcal{V}_{\mcal{S}_i,k}^i$, i.e., it makes $\eta_1\eta_2$ broadcast transmissions each of size $\frac{T}{r}$ bits. Each transmission is useful to all other $r$ nodes in $\mcal{S}$; moreover, each server in $\mcal{S}$ has the required side information to decode the part it needs.

\noindent 4) \emph{Reduce Phase}: In the reduce phase, every server uses its locally computed intermediate values and the decoded intermediate values in the shuffling phase to compute the $\eta_2$ output functions assigned to it in the initialization phase. 

Next we discuss the minimum computation load needed for the CDC scheme.
\medskip

\noindent {\bf Minimum Computations.}
The next proposition characterizes the minimum computation required by the CDC scheme.
\begin{prop}
    \label{prop:comp_CDC}
    For the placement scheme in \cite{li2016fundamental} with $r\!=\![1\!:\!K]$, the communication load $L^\star(r) = \frac{K-r}{rK}$ can be achieved with computation load
\begin{align}
    \mcal{C}_{\rm total} = \frac{rNQ(K-r+1)}{K}.
    \label{eq:total_comp_true}
\end{align}
\end{prop}
\begin{proof}
We first note that every server $i$ locally computes all intermediate values required by the functions in $\mcal{W}_i$ and corresponding to the files in $\mcal{M}_i$; we denote these intermediate values as $\mcal{C}_{\mcal{M}_i,\mcal{W}_i}$. 
Thus, we have $\mcal{C}_{\mcal{M}_i, \mcal{W}_i} = \{v_{q,n} | q \in \mcal{W}_i, n \in \mcal{M}_i \} \subseteq \mcal{C}_i$.
Note that $|\mcal{C}_{\mcal{M}_i, \mcal{W}_i} |{=} |\mcal{W}_i||\mcal{M}_i| {=} \eta_2 \frac{rN}{K}$.
In addition to $\mcal{C}_{\mcal{M}_i, \mcal{W}_i}$, server $i$ also performs a set of computations required to carry out shuffling in the CDC scheme. 
We denote this set by  $\mcal{C}_{TX_i}$.
To calculate the number of computations in $\mcal{C}_{TX_i}$,
we distinguish between computations required by server $i$ to decode its needed intermediate values (from transmissions in the shuffling phase) and the computations needed to create its transmissions $X_i$ in the shuffling phase. 

Observe (from the description of the CDC scheme earlier and in \cite{li2016fundamental}) that in any $\mcal{S} \subseteq [1\!:\!K]$ of size $r+1$ where $i \in \mcal{S}$, server $i$ uses the sets $\{\mcal{V}_{\mcal{S}_k,i}^k, | k \in \mcal{S}\backslash\{i\}\}$ to construct its transmission.
In addition, since the remaining parts $\{\mcal{V}_{\mcal{S}_k,j}^k | k \in \mcal{S}\backslash\{i\}, j \in \mcal{S}\backslash\{i,k\}\}$ will be XOR-ed (at the other servers) with parts needed by server $i$, then server $i$ should compute the intermediate values $\cup_{k \in \mcal{S}\backslash\{i\}} \mcal{V}_{\mcal{S}_k}^k$ in order to decode its requested intermediate values as well as construct its transmissions in the shuffling phase.
This amounts to $\sum_{k \in \mcal{S},k\neq i} |\mcal{V}_{\mcal{S}_k}^k| = r \eta_1 \eta_2$  computations for every set $\mcal{S}$. 
Thus, the total number of computations by server $i$, $|\mcal{C}_i|$, is
\begin{align}
|\mcal{C}_i| &=  |\mcal{C}_{TX_i}| + |\mcal{C}_{\mcal{M}_i,\mcal{W}_i}|\nonumber\\
&\stackrel{(i)}= {K{-}1 \choose r} r\eta_1 \eta_2 + \eta_2 \frac{r N}{K}\nonumber \stackrel{(ii)}= r\eta_2 \left( {K{-}1 \choose r}\eta_1 + \frac{N}{K} \right) \nonumber \\
&\stackrel{(iii)}= \frac{rQ}{K} \left( {K{-}1 \choose r}\frac{N}{{K \choose r}} {+} \frac{N}{K} \right) = \frac{rNQ(K{-}r{+}1)}{K^2},
    \label{eq:Comp_CDC_min}
\end{align}
where: (i) follows since server $i$ appears in only ${K{-}1 \choose r}$ subsets of size $r+1$; (ii) and (iii) follow from the assumptions that $N={K \choose r}\eta_1$ and $Q = K\eta_2$.
From symmetry, the total number of computations  in the Map phase equals $\mcal{C}_{\rm total}=K|\mcal{C}_i|$.
\end{proof}
Note from \eqref{eq:total_comp_true} that $C_{total}$ is quadratic in $r$. Thus, we cannot view $r$ as a direct measure of computation load since both the communication load $L$ as well as the number of computations $C_{\rm total}$ reduce for $r \geq (K+1)/2$. 
Fig.~\ref{fig:comp_comparison} shows the relation in \eqref{eq:total_comp_true} for $N = 2520$ and $K = Q = 10$ versus the number of computations if a server compute all map functions for each of its stored files.
If we use~\cite[Theorem 1]{li2016fundamental} and Proposition~\ref{prop:comp_CDC} to couple $C_{total}$ and $L^\star$, then we get the trade-off shown in Fig.~\ref{fig:tradeoff_modified} for the CDC scheme, where the red line is a scaled version of the trade-off in \cite{li2016fundamental}. 
From Fig.~\ref{fig:tradeoff_modified}, it can be seen that if we are free to choose $r$ for a given $C_{total}$, then the optimal trade-off happens at $C_{total} = NQ = 25200$; by picking $r=K=10$. This gives a communication load equal to zero while achieving the minimum computation load.
This observation suggests that we can better understand the communication-computation trade-off, if we consider it with a predefined redundancy load ($r$) that does not change with the computation load $C_{total}$.
\begin{figure*}
\centering
 \begin{minipage}{0.32\textwidth}
  \centering
  \includegraphics[width=0.98\textwidth]{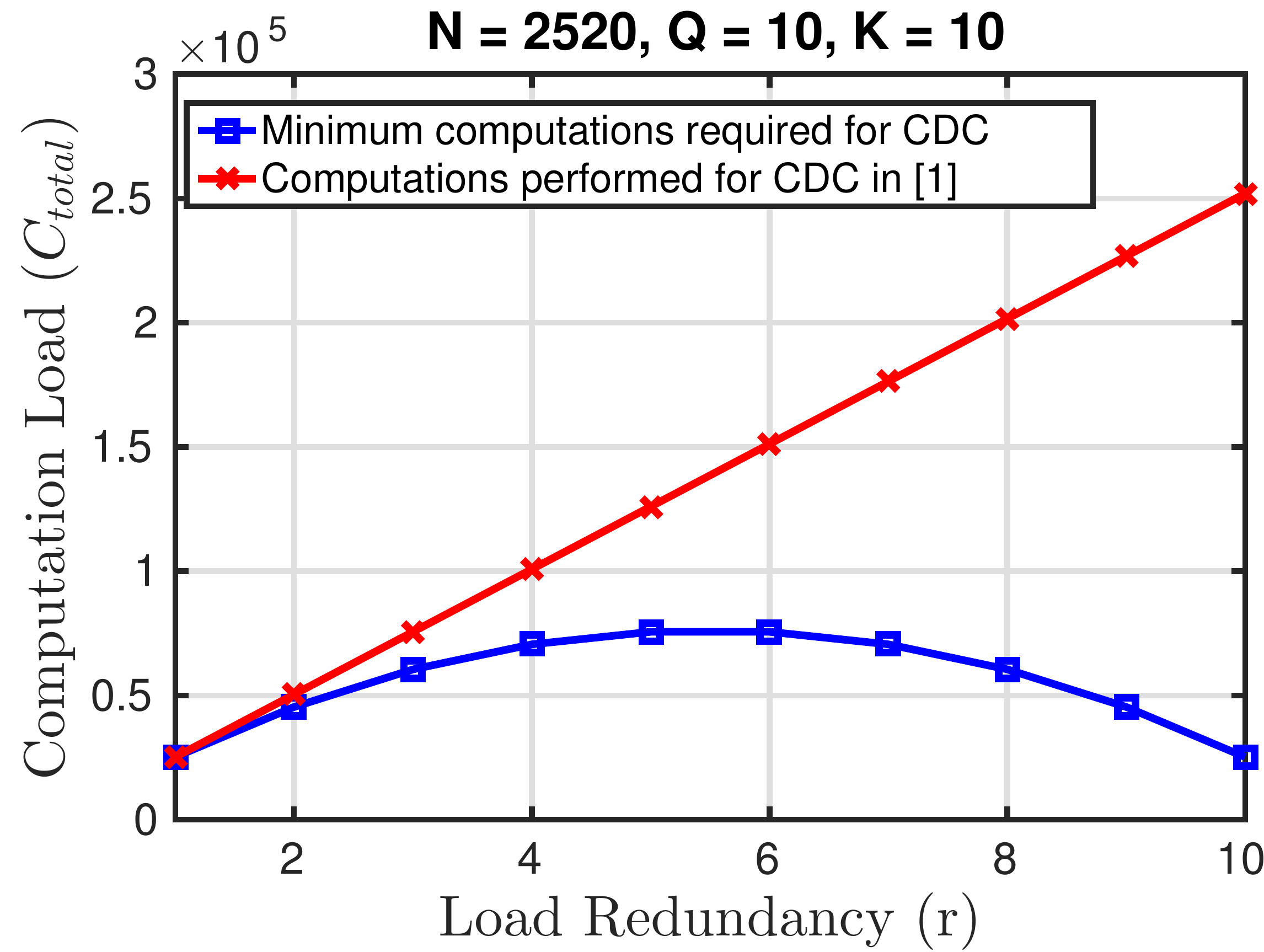}
    \caption{Computation Load vs. Load Redundancy.}
    \label{fig:comp_comparison}
 \end{minipage}
  \begin{minipage}{0.32\textwidth}
    \centering
    \includegraphics[width=1\textwidth]{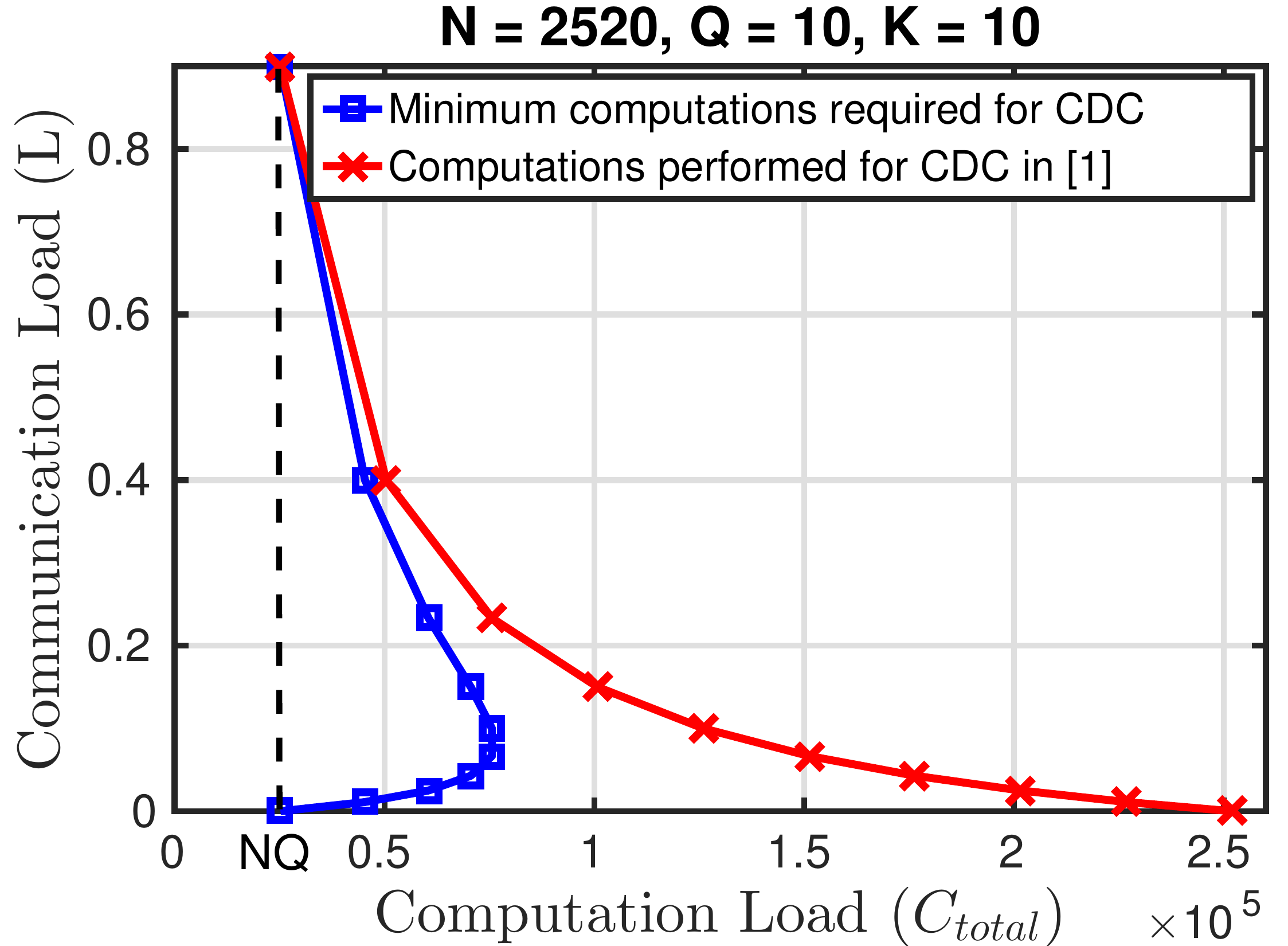}
    \caption{Communication Load vs. Computation Load.}
    \label{fig:tradeoff_modified}
 \end{minipage}
  \begin{minipage}{0.32\textwidth}
    \centering
    \includegraphics[width=1\textwidth]{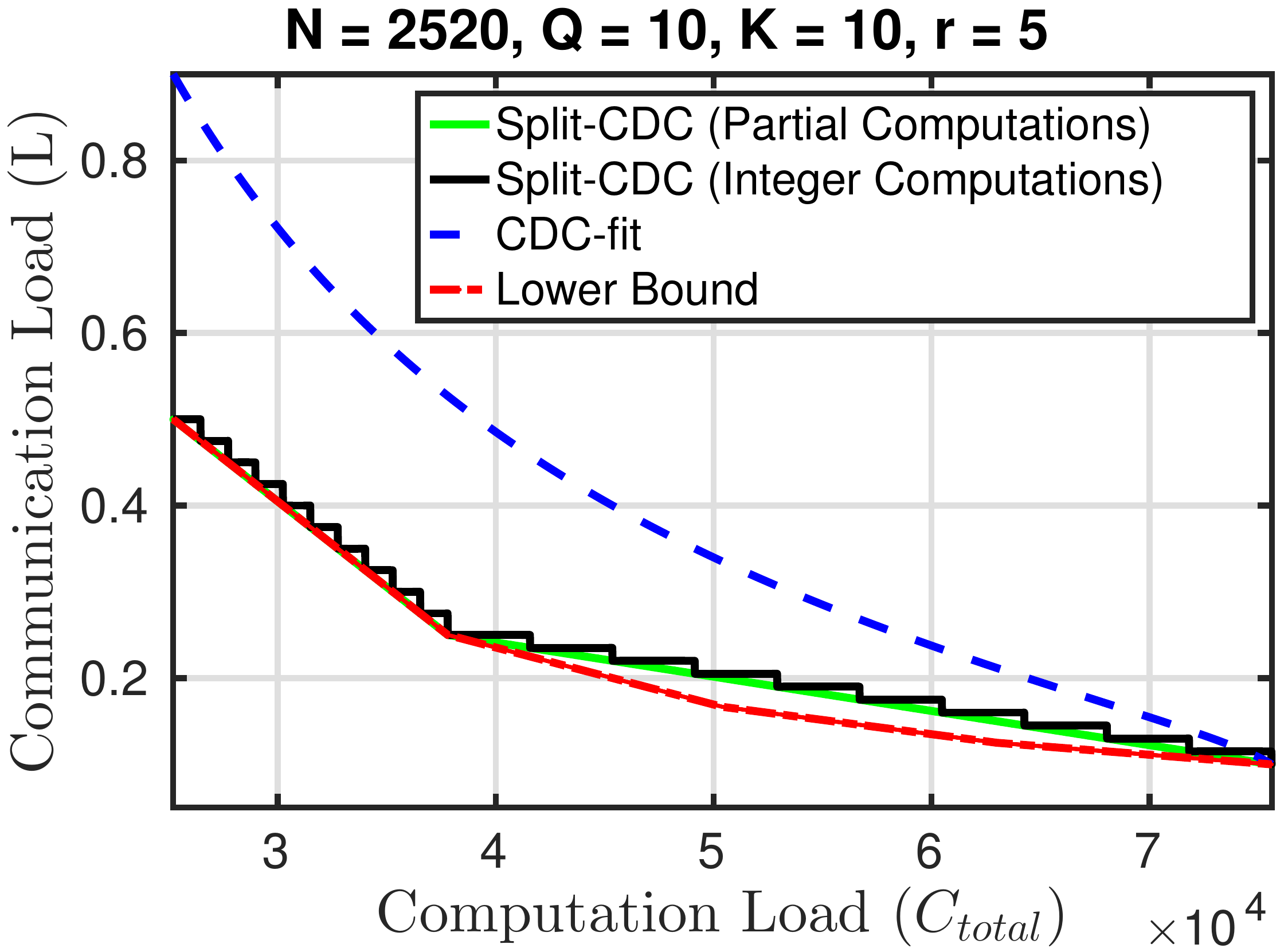}
    \caption{Communication Load vs. Computation Load for S-CDC.}
    \label{fig:LB_schemes}
 \end{minipage}
\vspace{-1em}
\end{figure*}

Thus, in the remainder of the paper, we consider $r$ as a parameter of the cluster (with $K, Q$ and $N$), and show how we can exploit this redundancy to perform coded distributed computing when at most $C_{total}$ computations are allowed.

\section{An Achievable Communication-Computation Trade-off}

Consider a distributed computing cluster with parameters $N,Q,K$ and load redundancy $r$, where $r$ represents {the number of times each} file is stored across the servers in the cluster. 
For our discussion in this section, we assume that $r \in [1:K]$ and that the file placement (for a given $r$) follows the strategy in~\cite{li2016fundamental}.
We are interested {in answering} the question: 
{If the cluster is allowed to perform at most $C_{total}$ computations, what is the minimum communication load $L(r,C_{total})$ needed in order to compute $Q$ output functions using the cluster ?}

If $C_{total} \!\geq\! r(K{-}r{+}1) NQ/K$, then from Proposition~\ref{prop:comp_CDC}, we can directly use the CDC scheme described in \cite{li2016fundamental}, to achieve the optimal communication load $L(r,C_{total}) = L^{\star}(r) = \frac{1}{r}\left(1-\frac{r}{K}\right)$. 
However, when $C_{total} <  r(K-r+1) NQ/K$, then the available computation budget is not enough to perform the shuffling and decoding required by the CDC scheme.
{In this case, can the CDC scheme be adapted to work with a restrictive computation budget?}
%
%
From~\cite{li2016fundamental}, we can infer a simple modification to the CDC scheme, which we refer to as {\it CDC-fit}. 
In this scheme, we use CDC on the cluster while operating it with a lower load redundancy $r$ that \emph{fits} the computation constraints.
In other words, we pick $r^\star {=} \max\{r' |C_{total} \!\geq\! r'{(K-r'+1)} NQ/K, r' \!\leq\! r\}$ and operate the cluster as if the files are only repeated $r^\star$ times. 
This ensures that there are enough computations to satisfy CDC for $r^\star$ and achieve the communication load $L(r^\star) = \frac{1}{r^\star}\left(1 - \frac{r^\star}{K}\right)$.
A natural question to ask here is whether this is the best possible approach?

To characterize this, we next develop a lower bound on the communication load when the cluster has a computation load $C_{total}$ and load redundancy $r$.

\smallskip

\noindent {\bf Lower Bound on Communication load.} 
We provide here a lower bound on the communication load for only a particular class of shuffling schemes. In this class, given a broadcast transmission sent during the shuffling phase, server $i$ can decode its required intermediate value from that transmission using {\it only} side information that it has {\it locally} computed.
i.e., it does not rely on future transmissions to provide it with enough linear combinations to decode its required intermediate values. In what follows, an \emph{$\ell$-type} transmission denotes a broadcast transmission made by a server during shuffling, which consists of the XOR of equally-sized parts of $\ell$ intermediate values. The weight of an $\ell$-type transmission is the size of the intermediate value parts used in the transmission.

    In order to relax our lower bound, we assume that a server can perform partial computations on the files, i.e., if a server wants to transmit a fraction of $fT$ bits (with $0 \leq f \leq 1$) of $v_{q,n}$ (recall $v_{q,n}$ is made of $T$ bits), then it only expends $f$ of a computation.
With this assumption, we can observe the following properties of our cluster: \\
\noindent\textbf{Obs. 1.} Each server has $rN/K$ files stored locally, and needs to receive $\frac{(K-r)N}{K} \cdot \frac{Q}{K}$ intermediate values through shuffling.\\
\noindent\textbf{Obs. 2.} For a cluster with load redundancy $r$, all feasible transmission have $\ell \!\leq\! r$. This follows by noting that an $\ell$-type transmission is assumed to satisfy $\ell$ servers~\footnote{If it is only useful for less than $\ell$ servers then the transmitter could have XOR-ed less intermediate values to generate the transmission.}. 
Therefore, each intermediate value involved in this transmission is computed once at the transmitter, and computed once at each of the other $\ell {-}1$ servers which would utilize this intermediate value as side information to decode the transmission. Since each file is repeated across $r$ servers, then $\ell \leq r$. \\
\noindent\textbf{Obs. 3.} In the shuffling phase, each $\ell$-type transmission and weight $fT$ incurs an added computation cost to the cluster equal to $\ell^2 fT$. To see this, note that the server sending this transmission makes $\ell fT$ computations. Moreover, an $\ell$-type transmission serves $\ell$ servers, each of which would have to do $(\ell -1)fT$ computations to acquire the needed side information. Therefore we get $\ell fT + \ell(\ell - 1)fT = \ell^2 fT$.
%

\smallskip

Let $z_\ell$ be the number of $\ell$-type transmissions.
Then, the communication load for a shuffling scheme is lower bounded by the solution of the following Linear Program (LP)
\begin{align}
&L_{lb}(C_{total})  = \min_{\{z_1,\dots,z_r\}} \sum_{\ell}^r \frac{z_\ell}{NQ} \label{eq:lowerbound_comm} \\
& s.t. \:\: \sum_{\ell=1}^r z_\ell\ \ell \geq \frac{(K-r)NQ}{K}, \:\:\: \sum_{\ell=1}^r z_\ell\ \ell^2 + \frac{rNQ}{K} \leq C_{total} \nonumber \\
 & \quad \:\:\:\:  z_i \geq 0,\quad \forall i \in [1:r], \nonumber
\end{align}
\noindent where: (i) the first constraint is a necessary condition for the shuffling phase to deliver $\frac{(K-r)QN}{K^2}$ intermediate values to each server in the cluster; 
(ii) the second condition is a necessary condition for the total computation (local computations and shuffling computations) to not exceed $C_{total}$.
Note that the result of the LP is a lower bound to the communication load 
since the first constraint is not sufficient to ensure that each server receives its needed intermediate values.
Fig~\ref{fig:LB_schemes} compares the communication-computation trade-off for the  aforementioned \emph{CDC-fit} scheme with the lower bound in \eqref{eq:lowerbound_comm}. The two trade-offs are close only towards high computation loads which allows the system to operate with an $r^\star$ close to the natural $r$ of the cluster.
Next, we propose a modification to the CDC scheme denoted as \emph{Split-CDC} (S-CDC) that provide a communication-computation trade-off close to the trade-off suggested by the lower bound in \eqref{eq:lowerbound_comm}. 

\smallskip

\noindent\textbf{Split-CDC (S-CDC).}
In order to introduce S-CDC, we make the following observations on the shuffling strategy in CDC.\\
\noindent\textbf{Obs. 1.} The set $\mcal{V}_{\mcal{S}_i}^i$ described in \eqref{eq:V_S_S1} is of size $|\mcal{V}_{\mcal{S}_i}^i| = \eta_1 \eta_2$. \\ 
\noindent\textbf{Obs. 2.} For every subset $\mcal{S}$ of $r+1$ servers, the computations needed to satisfy all servers in $\mcal{S}$ is $r(r{+}1)\eta_1 \eta_2$ and the number of packets communicated among them is $\frac{r+1}{r}\eta_1 \eta_2$.\\
\noindent\textbf{Obs. 3.} From \eqref{eq:V_S_S1}, it is not hard to see that for any subset $\mcal{S}' \subseteq \mcal{S}$ such that $|\mcal{S}'| > 1$, $\mcal{V}_{\mcal{S}_i}^i \subseteq \mcal{V}_{\mcal{S}'_i}^i, \forall i \in \mcal{S}'$. 

The previous observations suggest the following modification to the CDC scheme. Each subset $\mcal{S}$ of size $r+1$ can be {\it split} into disjoint subsets of smaller size. Each smaller subset $\mcal{S}'$ can still be used to satisfy its members with the set $\mcal{V}_{\mcal{S}_i}^i$ as per Observation $3$.
Therefore, by using subsets $\{\mcal{S}'\}$ of size different than $r+1$, this would allow the scheme to exhibit different levels of communications and computations per $\mcal{S}$ (based on the size of the splits), as evident from Observation $2$. 
The possible sizes of $\mcal{S}'$ are $r' \in [1:r]$, which we refer to as the {\it split size}. 
For $r'$, define $j_{r'} = \lfloor\frac{r+1}{r'+1}\rfloor$ and $r'' \!=\! (r\!+\!1) - j(r'\!+\!1)-1$. 
Thus we can split set $\mcal{S}$ into $j_{r'}$ disjoint sets $\mcal{S}^{(r')}$ of size $(r'+1)$ and one set $\mcal{S}^{(r'')}$ of size $(r''+1)$. 
For each set in $\mcal{S}^{(r')}$, the needed number of computations is $r'(r'+1)\eta_1\eta_2$ and the needed number of communicated packets is $\frac{r'+1}{r'}\eta_1\eta_2$. 
If $\mcal{S}^{(r'')}$ is not empty, then similar expression follow (except when $|\mcal{S}^{(r'')}| = 1$, where we need $\eta_1\eta_2$ computations and $\eta_1\eta_2$ packets exchanges to send the intermediate values through unicast transmissions from any server in $\mcal{S}^{(r')}$).
Finally, since $|\mcal{V}_{\mcal{S}_i}^i| = \eta_1 \eta_2$, for every subset $\mcal{S}$ of size $r+1$, CDC would naturally incur $\eta_1 \eta_2$ transmission rounds, each delivering exactly one intermediate value in $\mcal{V}_{\mcal{S}_i}^i$ for all servers in $\mcal{S}$.
Thus, our observations suggest that CDC can operate each of these transmission rounds with a different splitting size $r'$ of $\mcal{S}$;
thus the name {\it Split-CDC (S-CDC)}.
For a transmission round using split size $r'$, the total computations and communications per subset $\mcal{S}$ of size $r{+}1$ is
%
\begin{align*}
  Comp(r') &=
  \begin{cases}
      j_{r'} r'(r'+1) + r''(r''+1), & r'' \neq 0,\\
  j_{r'} r'(r'+1) + 1, & r'' = 0,
  \end{cases}\\
 Comm(r') &=
  \begin{cases}
  j_{r'}\frac{r'+1}{r'} + \frac{r''+1}{r''}, & r'' \neq 0,\\
  j_{r'}\frac{r'+1}{r'} + 1, & r'' = 0.
  \end{cases}
\end{align*}
    S-CDC can now be formally described. Let $\frac{z_{r'}}{\eta_1\eta_2}$ be the fraction of the intermediate values in $\mcal{V}_{\mcal{S}_i}^i$ per subset $\mcal{S}$ that is delivered using split size $r'$. Then, S-CDC works as follows:

\noindent$\bf 1)$ Determine the optimal values of $\frac{z_{r'}}{\eta_1\eta_2}$ for $r' \in [1:r]$ - this is done via solving the LP in \eqref{eq:p-cdc-scheme-comm}.\\
\noindent$\bf 2)$ For each $\mcal{S} \subseteq [1{:}K]$ of size $r{+}1$ and split size $r' \in [1{:}r]$:
\begin{itemize}
 \item Split set $\mcal{S}$ into $j_{r'}$ disjoint sets $\mcal{S}^{(r')}$ of size $(r'+1)$ and one set $\mcal{S}^{(r'')}$ of size $(r''+1)$.
 \item Use enough computations and communications per each of the subsets $\mcal{S}^{(r')}$ and $\mcal{S}^{(r'')}$ as per the CDC scheme, to deliver $z_{r'}$ intermediate values to all servers in $\mcal{S}$. The computations and communications needed to do so is equal to $z_{r'}Comp(r')$ and $z_{r'} Comm(r')$ respectively.
\end{itemize}
What remains is to find the optimal values of $z_{r'}$. We do so via solving the following LP, which minimizes the total communication load subject to a total computation constraint.
\begin{align}
    \label{eq:p-cdc-scheme-comm}
    &L_{P}(C_{total})\quad  =   {\min_{\{z_1,\dots,z_r\}}} {K \choose r+1}\sum_{\ell=1}^r \frac{z_{\ell} Comm(\ell)}{NQ} \nonumber\\
    & \quad \quad \: s.t.\ \ \sum_{\ell=1}^r z_\ell = \eta_1 \eta_2, \qquad z_i \geq 0,\quad \forall i \in [1:r], \nonumber \\
    & \quad \quad \: {K \choose r+1}\sum_{\ell=1}^r z_\ell Comp(\ell) + \frac{rNQ}{K} \leq C_{total}.
\end{align}

Note that in \eqref{eq:p-cdc-scheme-comm}, the variables $z_{\ell}$ are allowed to take non-integer values which means that we are allowing the servers to do partial computations of the intermediate values if that is what they will need to transmit or decode. 
To restrict partial computations, we can approximate the solution of \eqref{eq:p-cdc-scheme-comm} to get a suboptimal integer-valued solution $\hat{z}^\star_\ell$.
    Note that if an optimal solution of \eqref{eq:p-cdc-scheme-comm} is non-integer, then there exists only two non-zero elements of $\{z^\star_\ell\}$; 
        we denote these two elements as $z_{\ell_1}$ and $z_{\ell_2}$ where $\ell_1 < \ell_2$.
        Then for our approximate solution, we define $\hat{z}^\star_{\ell_2} = \lfloor {z}^\star_{\ell_2} \rfloor$ and $\hat{z}^\star_{\ell_1} = \eta_1\eta_2 - \lfloor{z}^\star_{\ell_2}\rfloor$.
        This gives us a communication load $\hat{L}_P(C_{total})= {K \choose r+1}\sum_{\ell=1}^r \frac{\hat{z}^\star_{\ell} Comm(\ell)}{NQ}$.

        Fig.~\ref{fig:LB_schemes} compares the performance of S-CDC with the lower bound in \eqref{eq:lowerbound_comm} for $N\!=\!2520, Q \!=\! K \!=\! 10$ and $r~=~5$ when partial computations are allowed. 
        In this particular setup, Fig.~\ref{fig:LB_schemes} shows that by preventing partial computations, we only incur a small fraction of the communication load as an expense.

\bibliographystyle{IEEEtran}
\bibliography{Distributed}
\end{document}